\documentclass[10pt,journal,compsoc]{IEEEtran}

\ifCLASSOPTIONcompsoc
\usepackage[nocompress]{cite}
\else
\usepackage{cite}
\fi

\ifCLASSINFOpdf
\else
\fi

\usepackage[ruled,vlined]{algorithm2e}
\usepackage{amsmath,empheq}
\usepackage{amssymb}
\usepackage{amsthm}
\usepackage{courier}
\usepackage{graphicx}
\usepackage{helvet}
\usepackage{enumitem}
\usepackage{multirow}
\usepackage{times}
\usepackage{balance}
\usepackage{cleveref}
\usepackage{url}

\usepackage{verbatim}
\usepackage{color}

\newtheorem{lemma}{Lemma}

\newtheorem{proposition}{Proposition}

\newcommand\T{\rule{0pt}{2ex}}
\newcommand\B{\rule[-1ex]{5pt}{0pt}}

\begin{document}

\title{\huge{\mbox{\!A Stochastic Game Framework for Analyzing Computational} Investment Strategies in Distributed Computing}
}

\markboth{Swapnil Dhamal et al. A Stochastic Game Framework for Analyzing Computational Investment Strategies in Distributed Computing
}{}

\author{Swapnil Dhamal, Walid Ben-Ameur, Tijani Chahed, Eitan Altman, Albert Sunny, and Sudheer Poojary
\IEEEcompsocitemizethanks{
\IEEEcompsocthanksitem Contact author: Swapnil Dhamal (swapnil.dhamal@gmail.com)
\vspace{1mm}
\IEEEcompsocthanksitem Swapnil  Dhamal is a postdoctoral researcher with Chalmers University of Technology, Sweden. A part of this work was done when he was a postdoctoral researcher with INRIA Sophia Antipolis-M\'editerran\'ee, France and T\'el\'ecom SudParis, France.
Walid Ben-Ameur and Tijani Chahed are professors with T\'el\'ecom SudParis, France.
Eitan Altman is a senior research scientist with INRIA Sophia Antipolis-M\'editerran\'ee, France.
Albert Sunny is an assistant professor with Indian Institute of Technology, Palakkad, India. A part of this work was done when he was a postdoctoral researcher with INRIA Sophia Antipolis-M\'editerran\'ee, France.
Sudheer Poojary is a senior lead engineer with Qualcomm India Pvt. Ltd. A part of this work was done when he was a postdoctoral researcher with Laboratoire Informatique d'Avignon, Universit\'e d'Avignon, France.
}
}

\IEEEtitleabstractindextext{
\begin{abstract}
We study a stochastic game framework with dynamic set of players, for modeling and analyzing their computational investment strategies in distributed computing. Players obtain a certain reward for solving the problem or for providing their computational resources, while incur a certain cost based on the invested time and computational power. We first study a scenario where the reward is offered for solving the problem, such as in blockchain mining. We show that, in Markov perfect equilibrium, players with cost parameters exceeding a certain threshold, do not invest; while those with cost parameters less than this threshold, invest maximal power. Here, players need not know the system state. We then consider a scenario where the reward is offered for contributing to the computational power of a common central entity, such as in volunteer computing. Here, in Markov perfect equilibrium, only players with cost parameters in a relatively low range in a given state, invest. For the case where players are homogeneous, they invest proportionally to the `reward to cost' ratio. For both the scenarios, we study the effects of players' arrival and departure rates on their utilities using simulations and provide additional insights.
\end{abstract}
}

\maketitle

\IEEEdisplaynontitleabstractindextext

\IEEEpeerreviewmaketitle

\IEEEraisesectionheading{\section{Introduction}}

\IEEEPARstart{D}{istributed}
computing systems comprise computers which coordinate
to solve large problems.
In a classical sense, a distributed computing system could be viewed as several providers of computational power contributing to the power of a common central entity (e.g. in volunteer computing~\cite{sarmenta2001volunteer,anderson2006computational}). The central entity  could, in turn, use the combined power for either fulfilling its own computational needs or distribute it to the next level of requesters of power (e.g. by a computing service provider to its customers in a utility computing model). The center would decide the time for which the system is to be run, and hence the compensation or reward to be given out per unit time to the providers. This compensation or reward would be distributed among the providers based on their respective contributions.
A provider incurs a certain  cost per unit time for investing a certain amount of  power.
So, in the most natural setting where the reward per unit time is distributed to the providers in proportion to their contributed power, a higher power investment by a provider is likely to fetch it a higher reward while also increasing its incurred cost, thus resulting in a tradeoff.

Distributed computing has gained more popularity 
than ever  
owing to the advent of blockchain.
Blockchain has found application in various fields \cite{zheng2018blockchain}, such as cryptocurrencies,
smart contracts, 
security services, 
public services, 
Internet of Things, etc.
Its functioning relies   on a proof-of-work procedure, where miners (providers of computational power)
collect block  data consisting of a number of transactions, and repeatedly compute hashes on inputs
from a very large search space. 
A miner is rewarded
for mining a block, if it finds one of the rare inputs that generates a hash value satisfying certain constraints,
before the other miners.
Given the cryptographic hash function, the best {known method for finding such an input is randomized search.}
Since the proof-of-work procedure is computationally  intensive, successful mining 
requires a miner to invest significant
computational power, resulting in the miner incurring some cost. 
Once a block is mined, it is transmitted to all the miners.
A miner's objective is to maximize its utility based on the offered reward for 
mining a block before others, by strategizing on the amount of  power to invest.
There is a natural tradeoff:
a higher  investment
increases a miner's chance of solving the problem before others, while
{a lower investment 
reduces its incurred cost.}

In this paper, we study the stochastic game where players (miners or providers of computational power) can arrive and depart during the mining of a block or during a run of volunteer computing.
We consider 
two of the most common scenarios in distributed computing, namely, (1) in which the reward is offered for solving the problem (such as in blockchain mining) and (2) in which the reward is offered for contributing to the computational power of a common central entity (such as in volunteer computing).

\vspace{-2mm}
\subsection{Preliminaries
}

\textbf{Stochastic Game.}
\cite{shapley1953stochastic}
It is a dynamic game with probabilistic transitions across different system states.
Players' payoffs and state transitions   depend on the current state and  players'  strategies. The game continues until it reaches a terminal state, if any.
Stochastic games are thus a generalization of both Markov decision processes and repeated games.

\vspace{1mm}
\noindent
\textbf{Markov Perfect Equilibrium (MPE).}
MPE \cite{maskin2001markov}
is an adaptation of subgame perfect Nash equilibrium to stochastic games.
An MPE strategy of a player is a policy function describing its strategy for each state, while ignoring history.
Each player computes its best response strategy in each state by foreseeing the effects of its actions on the state transitions and the resulting utilities, and the strategies of  other players.
A player's MPE policy is a best response to the other players' MPE policies.

It is worth noting that, while game theoretic solution concepts such as MPE, Nash equilibrium, etc. may seem impractical  owing to the common knowledge assumption, they provide a strategy profile  which can be suggested to  players (e.g. by a mediator) from which no player would  unilaterally deviate. Alternatively, if players play the game repeatedly while observing each other's actions, they would likely  settle at such a strategy profile.

\vspace{-2mm}
\subsection{Related Work
}

Stochastic games have been studied from  theoretical perspective~\cite{gillette1957stochastic,fink1964equilibrium,mertens1981stochastic,goeree1999stochastic,altman2006survey}
as well as in 
applications such as
computer networks~\cite{altman2003slotted},
cognitive radio networks~\cite{wang2011anti},
wireless network virtualization~\cite{fu2013stochastic},
queuing systems~\cite{altman1996non},
multiagent reinforcement learning~\cite{bowling2000analysis},
and complex living systems~\cite{bellomo2008modeling}.

{We enlist some of the important works on stochastic games.}
Altman and Shimkin
\cite{altman1998individual}
consider a processor-sharing  system, where an arriving customer  observes the current load on the shared system and chooses whether to join it or use a constant-cost alternative.
Nahir et al.
\cite{nahir2012workload}  study a similar setup, with the difference that 
customers consider using the
system over a long time scale and for multiple jobs.
Hassin and Haviv~\cite{hassin2002nash}
propose a version of subgame perfect Nash equilibrium for games where  players are identical; each player selects strategy based on its private information regarding the system  state.
Wang and Zhang~\cite{wang2013strategic}
investigate  Nash equilibrium 
in a queuing system, where reentering  the system is a strategic decision.
Hu and Wellman~\cite{hu2003nash} use the framework of general-sum stochastic games to extend Q-learning to a noncooperative multiagent context.
There   exist works which develop algorithms for computing  good, not necessarily optimal, strategies in a state-learning setting~\cite{jiang2014dynamic,wang2018game}.

Distributed systems have been studied from game theoretic perspective
in the literature~\cite{abraham2006distributed,kwok2005selfish}.
Wei et al.
\cite{wei2010game}
study a resource allocation game in a cloud-based network, with constraints on  quality of service.
Chun et al.~\cite{chun2004selfish}
analyze the selfish caching game, where  selfish server nodes incur either cost for replicating resources or cost for access to a remote replica.
Grosu and Chronopoulos
\cite{grosu2005noncooperative}
propose a game theoretic framework for obtaining a user-optimal load balancing scheme in heterogeneous distributed systems.

Zheng and Xie
\cite{zheng2018blockchain}
present a  survey on blockchain.
Sapirshtein et al.
\cite{sapirshtein2016optimal}
study selfish mining attacks, where a miner postpones  transmission of its mined blocks 
so as to prevent other miners from  starting the mining of the next block immediately.
Lewenberg et al.
\cite{lewenberg2015bitcoin}
study pooled mining, where miners form coalitions and share the obtained rewards,
so as to reduce the variance of the reward received by each player. 
Xiong et al.~\cite{xiong2018optimal}
consider that miners can offload
the mining process 
to an edge computing service provider. They study 
a Stackelberg game where
the provider sets  price for its services, and
the miners determine the
amount of services to  request.
Altman et al. 
\cite{altman2019mining} 
model the competition over several blockchains as a non-cooperative
game, and hence show the existence of pure Nash equilibria using a congestion game approach. 
Kiayias et al.~\cite{kiayias2016blockchain}
consider a stochastic game, where each state  corresponds to the mined blocks and the players who mined them;
players strategize on 
which blocks to mine
and
when to transmit them.

In general, there exist game theoretic studies for distributed systems, as well as stochastic games for  applications including blockchain mining (where a state, however, signifies the state of the chain 
of blocks).
To the best of our knowledge, 
this work is the first 
to study a stochastic game framework for distributed computing considering the set of players to be dynamic.
We consider the most general case of heterogeneous players; the cases of homogeneous  as well as multi-type players (which also have not been studied in the literature) are special cases of this study.

\vspace{-2mm}
\section{Our Model
}

Consider a distributed computing system wherein agents provide their computational power to the system, and receive a certain reward for successfully solving a problem or for providing their computational resources.
We first model the scenario where the reward is offered for solving the problem, such as in blockchain mining, and explain it in detail. We then model the scenario where the reward is offered for contributing to the computational power of a common central entity, such as in volunteer computing. We hence point out the similarities and differences between the utility functions of the players in the two scenarios.

\vspace{-2mm}
\subsection{Scenario 1: Model
}

We present our model for
blockchain mining, one of the most in-demand contemporary applications of the scenario where reward is offered for solving the problem. We conclude this subsection by showing that the utility function thus obtained, generalizes to other distributed computing applications belonging to this scenario.

Let $r$ be the  reward offered to a miner for successfully solving a problem, that is, for finding a solution before all the other miners.

\vspace{1mm}
\noindent
\textbf{Players.}
We consider that there are broadly two types of players (miners) in the system, namely, (a) strategic players who can arrive and depart while a problem is being solved (e.g., during the mining of a block) and can modulate the invested power based on the system state so as to maximize their expected reward and (b) fixed players who are constantly present in the system and invest a constant amount of power for large time durations (such as typical large mining firms).
In blockchain mining, for instance,
the universal set of players during the mining of a block consists of all those who are registered as miners at the time.
In particular, we denote by $\mathcal{U}$, the set of strategic players during the mining of the block under consideration.
We denote by $\ell$, the constant amount of power invested by the fixed players throughout the mining of the block under consideration.
We consider $\ell>0$ (which is true in actual mining owing to mining firms); so the mining does not stall even if the set of strategic players is empty. %
Since the fixed players are constantly present in the system and invest a constant amount of power, we denote them as a single aggregate player $k$, who invests a constant power of $\ell$ irrespective of the system state.

Since it may not be feasible for a player to manually modulate its invested power as and when the system changes its state, we consider that the power to be invested is  modulated by a pre-configured automated software running on the player's machine.
The player can strategically determine the policy, that is, how much to invest if the system is in a given state.

We denote by cost parameter $c_i$, the cost incurred by player $i$ for investing unit amount of  power for unit time.
We consider that  players are not constrained by the cost they could incur.
Instead, they aim to maximize their expected  utilities (the expected reward they would obtain minus the expected cost they would incur henceforth), while forgetting the cost they have incurred thus far. That is,  players are Markovian.
In our work, we assume that the cost parameters of all the players are common knowledge. This could be integrated in a blockchain mining or volunteer computing interface where players can declare their cost parameters. This information is then made available to the interfaces of all other players (that is, to the automated software running on the players' machines).
In real world, it may not be practical to make the players' cost parameters a common knowledge, and furthermore, players may not reveal them truthfully.
To account for such limitations, a mean field approach could be used by assuming homogeneous or multi-type players (which are special cases of our analysis). Furthermore, it is an interesting future direction to design incentives for the players to reveal their true costs.

\vspace{1mm}
\noindent
\textbf{Arrival and Departure of Players.}
For modeling the arrivals and departures of players,
we consider a standard queueing setting.
A player $j$, who is not in the system, arrives after time which is exponentially distributed with mean $1/\lambda_j$ (that is, the rate parameter is $\lambda_j$); this is in line with the Poission arrival process where the time for the first arrival is exponentially distributed with the rate parameter corresponding to the Poisson arrival.
Further, the departure time of a player $j$, who is in the system, is exponentially distributed with rate parameter $\mu_j$.
The stochastic arrival of players is natural, like in most applications. 
Further, players would usually shut down their computers on a regular basis, or terminate the computationally demanding mining task (by closing the automated software) so as to run other critical tasks. 
Note that since players are Markovian, they do not account for how much computation they have invested thus far for mining the current block. Also, as we shall later see, the computation itself is memoryless, that is, the time required to find the solution does not depend on the time invested thus far. Owing to these two reasons, the players do not monitor block mining progress, and hence depart stochastically.

\vspace{1mm}
\noindent
\textbf{State Space.}
Due to the arrivals and departures of strategic players,
{we could  view this  
as a continuous time multi-state process,} where a state corresponds to the  set of strategic players present in the system.
So, if the set of strategic players in the system is $S$ (which excludes the fixed players), we say that the system is in state $S$.
So, we have $S \subseteq \mathcal{U}$ or equivalently, $S \in 2^{\mathcal{U}}$.
In addition, we have $|\mathcal{U}|+1$ absorbing states corresponding to the problem being solved by the respective player (one of the strategic players in $\mathcal{U}$ or a fixed player). 
The players involved at any given time would influence each others' utilities, thus resulting in a game.
The stochastic arrival and departure of players makes it a stochastic game. As we will see, there are also other stochastic events in addition to the arrivals and departures, and which depend on the players' strategies.

\vspace{1mm}
\noindent
\textbf{Players' Strategies.}
Let $\tau=0$ denote the time when the mining of the current block begins.
Let $x_i^{(S,\tau)}$ denote the strategy of player $i$ (amount of  power  it decides to invest) at time $\tau$ when the system is in state $S$.
Since players use a randomized search approach over a search space which is exponentially large as compared to the solution space, the time required to find the solution 
is independent of the  search space explored thus far. That is, the search follows  memoryless property. 
Also, note that a player has no incentive to change its strategy amidst a state owing to this memoryless property and if no other player changes its strategy.
Hence in our analysis, we consider that no player changes its strategy within a state.
So we have $x_i^{(S,\tau)} = x_i^{(S,\tau')}$  for any $\tau,\tau'$; hence player $i$'s strategy could be written as a function of the state, that is, $x_i^{(S)}$.
For a state $S$ where $j \notin S$, we have $x_j^{(S)}=0$ by convention.
Let $\mathbf{x}^{(S)}$ denote the strategy profile of the players in state $S$.
Let $\mathbf{x} = (\mathbf{x}^{(S)})_{S \subseteq \mathcal{U}}$ denote the policy profile.

\vspace{1mm}
\noindent
\textbf{Rate of  Problem Getting Solved.}
As explained earlier, the time required to find a solution in a  large search space is independent of the  search space explored thus far. We consider  this time to be exponentially distributed to model its memoryless property
($\because$ if a continuous random variable has the memoryless property over the set of reals, it is necessarily exponentially distributed).
Let $\Gamma^{(S,\mathbf{x}^{(S)})}$ be the corresponding rate of problem getting solved in state $S$, when   players' strategy profile is $\mathbf{x}^{(S)}$.
Since the time required   is independent of the  search space explored thus far, the probability that a player  finds a solution before  others at time $\tau$ is proportional to its invested  power at time $\tau$.

Note that the time required for the problem to get solved is the minimum of the times required by the players to solve the problem.
Now, the minimum of exponentially distributed random variables, is another exponentially distributed random variable with rate which is the sum of the rates corresponding to the original random variables.
Furthermore, the probability of an original random variable being the minimum, is proportional to its rate.
Let $\text{P}_j^{(S,\mathbf{x}^{(S)})}$ be the rate (corresponding to an exponentially distributed random variable) of player $j$ solving the problem in state $S$, when the strategy profile is $\mathbf{x}^{(S)}$.
So, we have 
$
\Gamma^{(S,\mathbf{x}^{(S)})} = \sum_{j \in S \cup \{k\}} \text{P}_j^{(S,\mathbf{x}^{(S)})} 
$.
Since the probability that player $i$ solves the problem before the other players is proportional to its invested computational power at that time, we have that the rate of player $i$ solving the problem is
$\text{P}_i^{(S,\mathbf{x}^{(S)})} \!=\! \frac{x_i^{(S)}}{\sum_{j \in S} x_j^{(S)} + \ell} \Gamma^{(S,\mathbf{x}^{(S)})}$,
and the  rate of other players solving the problem is 
$\text{Q}_i^{(S,\mathbf{x}^{(S)})} \!=\! \sum_{j \in (S \setminus \{i\})\cup\{k\}} \text{P}_j^{(S,\mathbf{x}^{(S)})} \!=\! \frac{\sum_{j \in S\setminus\{i\}} x_j^{(S)}+\ell}{\sum_{j \in S} x_j^{(S)}+\ell} \Gamma^{(S,\mathbf{x}^{(S)})}$.

\setlength\tabcolsep{1mm}
\begin{table}
\caption{Notation}
\label{tab:notation}
\vspace{-1mm}
\begin{small}
\begin{tabular}{|c|p{8.4cm}|}
\hline \T \B
$r$ & reward parameter
\\ \hline \T \B
$c_i$ & cost incurred by player $i$ when it invests unit power for unit time
\\ \hline \T \B
$\lambda_i$ &  arrival rate corresponding to player $i$
\\ \hline \T \B
$\mu_i$ &  departure rate corresponding to player $i$
\\ \hline \T \B
$\mathcal{U}$ & universal set of strategic players  
\\ \hline \T \B
$\ell$ &  constant amount of power invested by the fixed players 
\\ \hline \T \B
$k$ &  aggregate player accounting for all the fixed players
\\ \hline \T \B
$S$ & set of strategic players  currently present in the system
\\ \hline \T \B
$x_i^{(S)}$ & strategy of player $i$ in state $S$
\\ \hline \T \B
$\mathbf{x}^{(S)}$ & strategy profile of players in state $S$
\\ \hline \T \B
$\mathbf{x}$ & policy profile 
\\ \hline \T \B
$\!\!\!\Gamma^{(S,\mathbf{x}^{(S)})}\!\!$ & rate of problem getting solved in state $S$ under strategy profile $\mathbf{x}^{(S)}\!\!$
\\ \hline \T \B
$R_i^{(S,\mathbf{x})}$ & expected utility of  $i$ computed in state $S$ under policy profile $\mathbf{x}$
\\ \hline
\end{tabular}
\end{small}
\vspace{-1mm}
\end{table}

\vspace{1mm}
\noindent
\textbf{The Continuous Time Markov Chain.}
Owing to the players being Markovian, 
when the system transits from state $S$ to state $S'$, each player $j \in S \cap S'$ could be  viewed as effectively reentering the system.
So, the 
expected utility could be written in a recursive form, which we now derive.
\Cref{tab:notation} presents the notation.
The possible events that can occur in a state $S \in 2^\mathcal{U}$ are:

\begin{enumerate}[leftmargin=*]
\setlength\itemsep{0em}
\item
the problem gets solved by player $i$ with rate $\text{P}_i^{(S,\mathbf{x}^{(S)})}$,
thus terminating the  game in the absorbing state where  $i$ gets a reward of $r$;

\item
the problem gets solved by one of the other players in $(S \setminus \{i\}) \cup \{k\}$ with rate $\text{Q}_i^{(S,\mathbf{x}^{(S)})}$,
thus terminating the game in an absorbing state  where player $i$ gets no reward;

\item
{a new player $j \in \mathcal{U} \setminus S$ arrives and the system transits to state $S \cup \{j\}$ with rate $\lambda_j$;}
\item
one of the players $j \in S$ departs and the system transits to state $S \setminus \{j\}$ with rate $\mu_j$.
\end{enumerate}

\noindent 
In what follows, we unambiguously write $j \in \mathcal{U} \setminus S$ as  \mbox{$j \notin S$}, for brevity.
{Since $\text{P}_i^{(S,\mathbf{x}^{(S)})} + \text{Q}_i^{(S,\mathbf{x}^{(S)})} = \Gamma^{(S,\mathbf{x}^{(S)})} $, the sojourn time in state $S$ is $( \Gamma^{(S,\mathbf{x}^{(S)})} + \sum_{j \notin S} \lambda_j + \sum_{j \in S} \mu_j )^{-\!1}\!$}.
Let $D^{(S,\mathbf{x})} = \Gamma^{(S,\mathbf{x}^{(S)})} + \sum_{j \notin S} \lambda_j + \sum_{j \in S} \mu_j $.
So, the 
expected cost incurred by player  $i$ while the system is in state $S$ is
$\frac{c_i x_i^{(S)}}{D^{(S,\mathbf{x})} }$.

\vspace{1mm}
\noindent
\textbf{Utility Function.}
The probability of an event occurring before any other event is equivalent to the corresponding exponentially distributed random variable being the minimum, which in turn, is proportional to its rate.
So, player $i$'s  expected utility  as computed in state $S$ is

\begin{small}
\begin{align}
\hspace{-2mm} 
R_i^{(S,\mathbf{x})}  = &
\frac{\Gamma^{(S,\mathbf{x}^{(S)})} \frac{x_i^{(S)}}{\sum_{j\in S} x_j^{(S)} +\ell}}{D^{(S,\mathbf{x})}} \!\cdot\! r
+
\frac{\Gamma^{(S,\mathbf{x}^{(S)})} \frac{\sum_{j\in S \setminus \{i\}} x_j^{(S)}+\ell}{\sum_{j\in S} x_j^{(S)}+\ell}}{D^{(S,\mathbf{x})}} \!\cdot\! 0
\nonumber
\\
& \hspace{-3mm}
\!+\! \sum_{j \notin S} \frac{\lambda_j}{D^{(S,\mathbf{x})}} \!\cdot\!  R_i^{(S \cup \{j\},\mathbf{x})} 
\!+\! \sum_{j \in S } \frac{\mu_j}{D^{(S,\mathbf{x})}} \!\cdot\!  R_i^{(S \setminus \{j\},\mathbf{x})} 
\!-\! \frac{c_i x_i^{(S)}}{D^{(S,\mathbf{x})}} 
\label{eqn:recursive} 
\end{align}
\end{small}

Note that 
we do not incorporate an explicit discounting factor with time. However, the utility of player $i$ can be viewed as discounting the future owing to the possibility that the problem can get solved in a state $S$ where $i \notin S$.
Moreover, our analyses are easily generalizable if an explicit discounting factor is incorporated.

For distributed computing applications with a fixed objective such as finding a solution to a given  problem, it is reasonable to assume that the rate of the problem getting solved is proportional to the total power invested by the providers of computation.
We, hence, consider that $\Gamma^{(S,\mathbf{x}^{(S)})} \!=\!  \gamma \left( \sum_{j \in S} x_j^{(S)} + \ell \right)$, where $\gamma$ is the rate constant 
of proportionality
determined by the problem being solved.
Hence, player $i$'s  expected utility  as computed in state $S$ is

\begin{small}
\begin{align}
\hspace{-2mm} 
\nonumber
R_i^{(S,\mathbf{x})}  = 
(\gamma r - c_i) \frac{ x_i^{(S)}}{D^{(S,\mathbf{x})}} 
& \!+\! \sum_{j \notin S} \frac{\lambda_j}{D^{(S,\mathbf{x})}} \!\cdot\!  R_i^{(S \cup \{j\},\mathbf{x})} 
\\ & \!+\! \sum_{j \in S } \frac{\mu_j}{D^{(S,\mathbf{x})}} \!\cdot\!  R_i^{(S \setminus \{j\},\mathbf{x})} 
\label{eqn:recursive_Sc1} 
\end{align}
\end{small}

\noindent
where $D^{(S,\mathbf{x})} = \gamma \left( \sum_{j \in S} x_j^{(S)} + \ell \right) + \sum_{j \notin S} \lambda_j + \sum_{j \in S} \mu_j$.

\vspace{1mm}
\noindent
\textbf{Other Applications of Scenario 1.}
We derived Expression~(\ref{eqn:recursive}) for the expected utility by considering that the probability of player $i$ being the first to solve the problem is proportional to its invested  power at the time, and hence obtains the reward $r$ with this probability.
Now, consider another type of system which aims to solve an NP-hard problem where players search for a solution, and the system rewards the players in proportion to their invested power when the problem gets solved. In this case, the first two terms of Expression~(\ref{eqn:recursive}) are replaced with the term
$\frac{\Gamma^{(S,\mathbf{x}^{(S)})} \left(\frac{x_i^{(S)}}{\sum_{j\in S} x_j^{(S)} +\ell}  r \right)}{D^{(S,\mathbf{x})}} $.
So,  the  mathematical form stays the same, and so when $\Gamma^{(S,\mathbf{x}^{(S)})} \!=\!  \gamma \left( \sum_{j \in S} x_j^{(S)} + \ell \right)$, our analysis presented in \Cref{sec:scenario1} holds for this case too.

\vspace{-2mm}
\subsection{Scenario 2: Model
}

We now consider the scenario 
where the reward is offered for contributing to the computational power of a common central entity, such as in volunteer computing.
Here, the reward offered per unit time is inversely proportional to the expected time for which the center decides to run the system.
Considering that the time for which the center plans to run the system is exponentially distributed with rate parameter $\beta$, the reward offered per unit time is inversely proportional to  $\frac{1}{\beta}$, and hence directly proportional to $\beta$.
Hence, let the offered reward per unit time be $r \beta$,  where $r$ is the reward constant of proportionality.
Furthermore, the reward given to a player is proportional to its computational investment.
So, the revenue received by player $i$ per unit time is $\frac{x_i^{(S)}}{\sum_{j\in S} x_j^{(S)} +\ell}  r \beta$, and hence its net profit per unit time is $\frac{x_i^{(S)}}{\sum_{j\in S} x_j^{(S)} +\ell}  r \beta - c_i x_i^{(S)}$.
The sojourn time in state $S$, similar to the previous scenario, is $\frac{1}{D^{(S,\mathbf{x})}}$, where $D^{(S,\mathbf{x})} = \beta + \sum_{j \notin S} \lambda_j + \sum_{j \in S} \mu_j$ (here, we have $\beta$ instead of $\Gamma^{(S,\mathbf{x}^{(S)})}$). So, the net expected  profit made by player $i$ in state $S$ before the system transits to another state, is $\frac{\frac{x_i^{(S)}}{\sum_{j\in S} x_j^{(S)} +\ell} r \beta - c_i x_i^{(S)}}{D^{(S,\mathbf{x})}} $.

Hence, player $i$'s expected utility as computed in state $S$ is

\begin{small}
	\begin{align}
	\nonumber
	\hspace{-2mm} 
	R_i^{(S,\mathbf{x})}  = 
	\frac{\frac{x_i^{(S)}}{\sum_{j\in S} x_j^{(S)} +\ell} r \beta - c_i x_i^{(S)}}{D^{(S,\mathbf{x})}} 
	& \!+\! \sum_{j \notin S} \frac{\lambda_j}{D^{(S,\mathbf{x})}} \!\cdot\!  R_i^{(S \cup \{j\},\mathbf{x})} 
	\\ & \!+\! \sum_{j \in S } \frac{\mu_j}{D^{(S,\mathbf{x})}} \!\cdot\!  R_i^{(S \setminus \{j\},\mathbf{x})} 
\label{eqn:recursive_Sc2} 
	\end{align}
\end{small}

Note that since $D^{(S,\mathbf{x})} = \beta + \sum_{j \notin S} \lambda_j + \sum_{j \in S} \mu_j$ here, Expression~(\ref{eqn:recursive_Sc2}) is obtainable from Expression~(\ref{eqn:recursive}), when 
$\Gamma^{(S,\mathbf{x}^{(S)})} = \beta$.

\vspace{1mm}
\noindent
\textbf{Other Variants of Scenario 2.}
We considered that the time for which the center decides to run the system is exponentially distributed with rate parameter $\beta$, where $\beta$ is a constant.
For theoretical interest, one could consider a generalization where
the system may dynamically determine this parameter based on the set of players $S \cup \{k\}$ present in the system.
Let such a rate parameter be given by $f(S)$. Since the fixed players and their invested power do not change, these could be encoded in $f(\cdot)$, thus making it a function of only the set of strategic players.
The center could determine $f(S)$
based on the cost parameters of the players in set $S$, the past records of the investments of players in set $S$, etc.
If the time for which the system is to run is independent of the set of players currently present in the system, we have the special case: $f(S) = \beta, \forall S$.
It can be easily seen that the analysis presented in this paper (\Cref{sec:scenario2}) goes through directly by replacing $\beta$ with $f(S)$, since $\Gamma^{(S,\mathbf{x}^{(S)})} = f(S)$ is also independent of the players' investment strategies.

Further, note that if the rate parameter is not just dependent on the set of players present in the system but also proportional to their invested power, it could be written as $\Gamma^{(S,\mathbf{x}^{(S)})} \!=\!  \gamma \left( \sum_{j \in S} x_j^{(S)} + \ell \right)$. This leads to the utility function being given by Equation~(\ref{eqn:recursive_Sc1}) and hence its analysis is same as that of Scenario 1 (\Cref{sec:scenario1}).

\subsection*{Convergence of Expected Utility
}

Note that Equation~(\ref{eqn:recursive}) encompasses both scenarios, where $\Gamma^{(S,\mathbf{x}^{(S)})} = \gamma \left( \sum_{j \in S} x_j^{(S)} + \ell \right)$ leads to Scenario 1, while $\Gamma^{(S,\mathbf{x}^{(S)})} = \beta$ leads to Scenario 2.
We now show the convergence of this recursive equation, and hence derive a closed-form expression for  utility function.

Let us define an ordering $\mathcal{O}$ on sets  which presents a one-to-one mapping from a set $S \subseteq \mathcal{U}$ to an integer between $1$ and $2^{|\mathcal{U}|}$, both inclusive.
Let $\mathbf{R}_i^{(\mathbf{x})}$ be the vector whose  component $\mathcal{O}(S)$ is $R_i^{(S,\mathbf{x})}$. We now show that $\mathbf{R}_i^{(\mathbf{x})}$ computed  using   the recursive Equation~(\ref{eqn:recursive}), converges for any policy profile $\mathbf{x}$.
Let $\mathbf{W}^{(\mathbf{x})}$ be the state transition matrix, among the states corresponding to the  set of strategic players present in the system.
In what follows, instead of writing  $W^{(\mathbf{x})}(\mathcal{O}(S),\mathcal{O}(S'))$, we simply write $W^{(\mathbf{x})}(S,S')$
since it does not introduce any ambiguity.
So, the elements of $\mathbf{W}^{(\mathbf{x})}$ are as follows:

\vspace{-3mm}
\begin{small}
\begin{align*}
\hspace{-3mm}
\text{For $j \notin S$}:\,
& W_i^{(\mathbf{x})}(S,S \cup \{j\}) = \frac{\lambda_j}{D^{(S,\mathbf{x})}} \\ 
\text{For $j \in S$}:\,
& W_i^{(\mathbf{x})}(S,S \setminus \{j\}) = \frac{\mu_j}{D^{(S,\mathbf{x})}} \;,
\end{align*}
\end{small}
\vspace{-1mm}

\noindent
All other elements of $\mathbf{W}^{(\mathbf{x})}\!$ are $0$.

Since $\ell>0$, we have that $\Gamma^{(S,\mathbf{x}^{(S)})} > 0$.
So, $D^{(S,\mathbf{x})} > \sum_{j \notin S} \lambda_j + \sum_{j \in S} \mu_j $.
Hence, $\!\mathbf{W}_i^{(\mathbf{x})}\!$ is  strictly substochastic  (sum of the elements in each of its rows is less than 1).
\\
Let $\mathbf{Z}_i^{(\mathbf{x})}$ be the vector whose  component $\mathcal{O}(S)$ is $Z_i^{(S,\mathbf{x})}\!$, where

\vspace{-1mm}
\begin{small}
\begin{align*}
& Z_i^{(S,\mathbf{x})} \!=\! \Bigg( \frac{\Gamma^{(S,\mathbf{x}^{(S)})}}{\sum_{j\in S} x_j^{(S)}+\ell} r \!-\!c_i \Bigg)\frac{x_i^{(S)}}{D^{(S,\mathbf{x})}} \;,
\end{align*}
\end{small} 

\begin{proposition}
$\mathbf{R}_i^{(\mathbf{x})} = (\mathbf{I}-\mathbf{W}^{(\mathbf{x})})^{-1} \mathbf{Z}_i^{(\mathbf{x})}$.
\label{prop:convergence}
\end{proposition}

\begin{proof}
	
	Let $\mathbf{R}_{i\langle t \rangle}^{(\mathbf{x})} = (R_{i \langle t \rangle}^{(1,\mathbf{x})},\ldots,R_{i \langle t \rangle}^{(2^{|\mathcal{U}|},\mathbf{x})})^T$,
	where $t$ is the iteration number and $(\cdot)^T$ stands for matrix transpose.
	The iteration for the value of $\mathbf{R}_{i\langle t \rangle}^{(\mathbf{x})}$ starts at $t=0$; we examine if it converges when $t \rightarrow \infty$. Now, the expression for the expected utility in all states can be written in matrix form and then solving the recursion, as
	
	\vspace{-3mm}
	\begin{small}
		\begin{align*}
		\mathbf{R}_{i \langle t \rangle}^{(\mathbf{x})} &= \mathbf{W}^{(\mathbf{x})}  \mathbf{R}_{i \langle t-1 \rangle}^{(\mathbf{x})} + \mathbf{Z}_i^{(\mathbf{x})}
	\\	&= \left(\mathbf{W}^{(\mathbf{x})}\right)^t \mathbf{R}_{i \langle 0 \rangle}^{(\mathbf{x})} + \Bigg( \sum_{\eta=0}^{t-1} \left(\mathbf{W}^{(\mathbf{x})}\right)^\eta \Bigg) \mathbf{Z}_i^{(\mathbf{x})}
		\end{align*}
	\end{small}
	
	\noindent
	Now, since $\mathbf{W}^{(\mathbf{x})}$ is  strictly substochastic,
	its spectral radius is  less than 1.
	So when $t \rightarrow \infty$, we have $\lim_{t \rightarrow \infty} (\mathbf{W}^{(\mathbf{x})})^t\! = \mathbf{0}$.
	Since $\mathbf{R}_{i \langle 0 \rangle}^{(\mathbf{x})}$ is a finite constant,
	we have $\lim_{t \rightarrow \infty} (\mathbf{W}^{(\mathbf{x})})^t \mathbf{R}_{i \langle 0 \rangle}^{(\mathbf{x})} = \mathbf{0}$.
	Further,
	$\lim_{t \rightarrow \infty} \sum_{\eta=0}^{t-1} (\mathbf{W}^{(\mathbf{x})})^\eta = (\mathbf{I}-\mathbf{W}^{(\mathbf{x})})^{-1}$
	\cite{hubbard2015vector}. 
	This  implicitly means that $(\mathbf{I}-\mathbf{W}^{(\mathbf{x})})$ is invertible.  
	Hence, 
	
	\vspace{-2mm}
	\begin{small}
		\begin{align*}
		\lim_{t \rightarrow \infty} \mathbf{R}_{i \langle t \rangle}^{(\mathbf{x})} &= \lim_{t \rightarrow \infty} \left(\mathbf{W}^{(\mathbf{x})}\right)^t \mathbf{R}_{i \langle 0 \rangle}^{(\mathbf{x})} + \Bigg( \sum_{\eta=0}^{\infty} \left(\mathbf{W}^{(\mathbf{x})}\right)^\eta \Bigg) \mathbf{Z}_i^{(\mathbf{x})}
		\\[0em]
		&= \mathbf{0} + (\mathbf{I}-\mathbf{W}^{(\mathbf{x})})^{-1} \mathbf{Z}_i^{(\mathbf{x})}
		\qedhere
		\end{align*}
	\end{small}
\end{proof}

\vspace{2mm}
Owing to the requirement of deriving the inverse of $\mathbf{I}-\mathbf{W}^{(\mathbf{x})}$,
it is clear that a general analysis of the concerned stochastic game when considering an arbitrary $\mathbf{W}^{(\mathbf{x})}$ is intractable.
In this work, we consider two special scenarios that we motivated earlier in the context of distributed computing systems, 
for which we show that the analysis turns out to be tractable.

\vspace{-2mm}
\section{Scenario 1: Analysis of MPE
}
\label{sec:scenario1}

Let $\hat{R}_i^{(S,\mathbf{x})}$ be the equilibrium utility of player $i$ in state $S$, that is, when  $i$ plays its best response strategy to the equilibrium strategies of the other players $j \in S \setminus \{i\}$ (while foreseeing   effects of its actions on  state transitions and resulting utilities). 
We can determine MPE similar to  optimal policy in MDP (using policy-value iterations to reach a fixed point).
Here, for maximizing $\hat{R}_i^{(S,\mathbf{x})}$,
we could assume that we have optimized for other states and use those values to find an optimizing $\mathbf{x}$ for maximizing $\hat{R}_i^{(S,\mathbf{x})}$.
In our case, we have  a closed form expression for  vector $\mathbf{R}_i^{(\mathbf{x})}$ in terms of policy $\mathbf{x}$ (\Cref{prop:convergence}); so we could effectively determine the fixed point directly.

A  policy is said to be {\em proper\/} if 
from any initial
state, the probability of reaching a terminal state 
is strictly positive.
Consider the condition that, there exists at least one proper policy, and for any non-proper policy, there exists at least
one state where the 
value function is negatively unbounded.
It is known that, under this condition, the optimal value
function is bounded, and it is the unique fixed point of the optimal Bellman
operator~\cite{bertsekas1995neuro}.
Our model satisfies this condition,
since there does not exist any non-proper policy 
as the probability of reaching a terminal state  corresponding to the problem getting solved (either by player $i$ or any other player including the fixed players) is strictly positive ($\because \Gamma^{(S,\mathbf{x}^{(S)})} > 0$).

Now, from Equation~(\ref{eqn:recursive_Sc1}), the Bellman equations over states $S \in 2^\mathcal{U}$ for player $i$ can be written as

\vspace{-2mm}
\begin{small}
\begin{align*} 
\hat{R}_i^{(S,\mathbf{x})}  = \max_{\mathbf{x}} \Bigg\{ (\gamma r - c_i)
\frac{x_i^{(S)}}{D^{(S,\mathbf{x})}} 
& \!+\! \sum_{j \notin S} \frac{\lambda_j}{D^{(S,\mathbf{x})}} \!\cdot\! \hat{R}_i^{(S\cup \{j\},\mathbf{x})} 
\\ & \!+\! \sum_{j\in S} \frac{\mu_j}{D^{(S,\mathbf{x})}} \!\cdot\! \hat{R}_i^{(S\setminus \{j\},\mathbf{x})} 
\Bigg\}
\end{align*}
\end{small}
\vspace{-.5mm}

We now derive some results, leading to the derivation of MPE.

\begin{lemma}
In Scenario 1,
for any state $S$ and  policy profile $\mathbf{x}$, we have
$R_i^{(S,\mathbf{x})} \!<\! r \!-\! \frac{c_i}{\gamma}$ if $\gamma r \!>\! c_i$, and  $R_i^{(S,\mathbf{x})} \!>\! r \!-\! \frac{c_i}{\gamma}$ if $\gamma r \!<\! c_i$.
\label{lem:bounded}
\end{lemma}
\begin{proof}

Let ${V}_i^{(S,\mathbf{x}^{(S)})}$ be the expected utility of player $i$ in state $S$ computed without considering the arrivals and departures of players ($\lambda_j=0, \forall j \notin S$ and $\mu_j=0, \forall j \in S$). So, we have

\vspace{-2mm}
\begin{small}
\begin{align*}
V_i^{(S,\mathbf{x}^{(S)})}
=
(\gamma r \!-\! c_i)  \frac{x_i^{(S)}}{\gamma \left( \sum_{j \in S} x_j^{(S)} + \ell \right)}
= \left(  r \!-\! \frac{c_i}{\gamma} \right) \! \frac{x_i^{(S)}}{  \sum_{j \in S} x_j^{(S)} + \ell }
\end{align*}
\end{small} 

Let $\mathbf{V}_i^{(\mathbf{x})}$ be the vector whose  component   $\mathcal{O}(S)$ is ${V}_i^{(S,\mathbf{x}^{(S)})}$.
Let $\mathbf{Z}_i^{(\mathbf{x})} = \mathbf{Y}^{(\mathbf{x})} \mathbf{V}_i^{(\mathbf{x})}$.
Note that when $\Gamma^{(S,\mathbf{x}^{(S)})} =  \gamma \left( \sum_{j \in S} x_j^{(S)} +\ell \right)$, we have that $\mathbf{Y}^{(\mathbf{x})}$ is a diagonal matrix, with diagonal elements
${Y}^{(\mathbf{x})}(S,S) = 
\frac{\gamma \left( \sum_{j\in S} x_j^{(S)} +\ell \right)}{D^{(S,\mathbf{x})}}
=
\frac{\gamma \left( \sum_{j\in S} x_j^{(S)} +\ell \right)}{\sum_{j \notin S} \lambda_j + \sum_{j\in S} \mu_j + \gamma \left( \sum_{j\in S} x_j^{(S)} +\ell \right)}$.
It can hence be seen that
$\mathbf{W}^{(\mathbf{x})} + \mathbf{Y}^{(\mathbf{x})}$ is a stochastic matrix (the sum of elements in each of its rows is 1).

Let $\mathbf{U}^{(\mathbf{x})} = (\mathbf{I}-\mathbf{W}^{(\mathbf{x})})^{-1} \mathbf{Y}^{(\mathbf{x})} \mathbf{1}$, where $\mathbf{1}$ is the vector whose each element is $1$.
It is clear that all the elements of $\mathbf{U}^{(\mathbf{x})}$ are non-negative.
We will now show that $||\mathbf{U}^{(\mathbf{x})}||_\infty \leq 1$, that is, the maximum element of the vector $\mathbf{U}^{(\mathbf{x})}$ is not more than 1.
Let $u_{S_0}$ be the element with the maximum value (one of the maximum, if there are multiple).
Suppose $u_{S_0}^{(\mathbf{x})} = ||\mathbf{U}^{(\mathbf{x})}||_\infty > 1$.
So, we would have

~
\vspace{-3mm}
\begin{small}
\begin{align*}
& \mathbf{U}^{(\mathbf{x})} = (\mathbf{I}-\mathbf{W}^{(\mathbf{x})})^{-1} \mathbf{Y}^{(\mathbf{x})} \mathbf{1}
\\
\implies &
\mathbf{U}^{(\mathbf{x})} = \mathbf{W}^{(\mathbf{x})} \mathbf{U}^{(\mathbf{x})} + \mathbf{Y}^{(\mathbf{x})} \mathbf{1}
\\
\implies
& u_{S_0}^{(\mathbf{x})} = \sum_{S \in 2^{\mathcal{U}}} u_{S}^{(\mathbf{x})} W^{(\mathbf{x})}(S_0,S) + {Y}^{(\mathbf{x})} (S_0,S_0)
\\
\implies
& u_{S_0}^{(\mathbf{x})} < u_{S_0}^{(\mathbf{x})} \sum_{S \in 2^{\mathcal{U}}} W^{(\mathbf{x})}(S_0,S) + u_{S_0}^{(\mathbf{x})} {Y}^{(\mathbf{x})} (S_0,S_0)
\\
&\;\;\;\;\;\;\;\;\;\;\;\;\;\;\;\;\;\;\;\;\;\;\;\;\;\;\;\;\;\;\;\;\;\;\;\;\;\;\;\;\;[\because  \max_S u_{S}^{(\mathbf{x})} = u_{S_0}^{(\mathbf{x})} > 1]
\\
\implies
& \sum_{S \in 2^{\mathcal{U}}} W^{(\mathbf{x})}(S_0,S) + {Y}^{(\mathbf{x})} (S_0,S_0) > 1
\end{align*}
\end{small}

However, this is a contradiction since $\mathbf{W}^{(\mathbf{x})} + \mathbf{Y}^{(\mathbf{x})}$ is a stochastic matrix.
So, we have shown that 
$||\mathbf{U}^{(\mathbf{x})}||_\infty = ||(\mathbf{I}-\mathbf{W}^{(\mathbf{x})})^{-1} \mathbf{Y}^{(\mathbf{x})} \mathbf{1}||_\infty \leq 1$.
That is, $(\mathbf{I}-\mathbf{W}^{(\mathbf{x})})^{-1} \mathbf{Y}^{(\mathbf{x})}$ is either stochastic or substochastic.

\noindent
From \Cref{prop:convergence}, 
$\mathbf{R}_i^{(\mathbf{x})} \!=\! (\mathbf{I}-\mathbf{W}^{(\mathbf{x})})^{-1} \mathbf{Y}^{(\mathbf{x})} \mathbf{V}_i^{(\mathbf{x})}$.
Since $(\mathbf{I}-\mathbf{W}^{(\mathbf{x})})^{-1} \mathbf{Y}^{(\mathbf{x})}$ is stochastic or substochastic, $R_i^{(S,\mathbf{x})}$ for each $S$ is a linear combination (with weights summing to less than or equal to 1) of ${V}_i^{(S,\mathbf{x}^{(S)})}$ over all $S \in 2^\mathcal{U}$.

For each $S$,
$V_i^{(S,\mathbf{x}^{(S)})} = \left(  r - \frac{c_i}{\gamma} \right) \! \frac{x_i^{(S)}}{\sum_{j\in S} x_j^{(S)}+\ell}$. So we have 
$V_i^{(S,\mathbf{x}^{(S)})} \!<\! r - \frac{c_i}{\gamma} \text{ if } \gamma  r \!>\! c_i, \text{ and } V_i^{(S,\mathbf{x}^{(S)})} \!>\! r - \frac{c_i}{\gamma} \text{ if } \gamma  r \!<\! c_i$.
Since
$R_i^{(S,\mathbf{x})}$ for each $S$ is a linear combination (with  weights summing to less than or equal to 1) of ${V}_i^{(S,\mathbf{x}^{(S)})}$ over all $S \in 2^\mathcal{U}$,
we have  
$R_i^{(S,\mathbf{x})} \!<\! r - \frac{c_i}{\gamma}$ if $\gamma r \!>\! c_i$, and  $R_i^{(S,\mathbf{x})} \!>\! r - \frac{c_i}{\gamma}$ if $\gamma r \!<\! c_i$.
\end{proof}

\begin{lemma}
{In Scenario$\,$1, $\!R_i^{(S,\mathbf{x})}\!\!$ is a monotone function of $x_i^{(S)}\!\!$.}
\label{lem:monotone}
\end{lemma}

\begin{proof}

We define the following for simplifying  notation.

\vspace{-2mm}
\begin{small}
\begin{align*}
A_i^{(S,\mathbf{x})} = \sum_{j\notin S} \lambda_j \hat{R}_i^{(S \cup \{j\},\mathbf{x})} \!+\! \sum_{j\in S} \mu_j \hat{R}_i^{(S\setminus \{j\},\mathbf{x})}
\end{align*}
\end{small}
\vspace{-5mm}

\begin{small}
\begin{align*}
E_i^{(S,\mathbf{x}^{(S)})} = \sum_{j\notin S} \lambda_j  \!+\! \sum_{j\in S} \mu_j \!+\! \gamma \Bigg( \! \sum_{j\in S\setminus \{i\}} x_j^{(S)} +\ell \Bigg)
\end{align*}
\end{small}

\noindent
Hence, we can write

\vspace{-2mm}
\begin{small}
\begin{align*}
R_i^{(S,\mathbf{x})} = \frac{A_i^{(S,\mathbf{x})} + (\gamma r - c_i) x_i^{(S)}}{E_i^{(S,\mathbf{x}^{(S)})} + \gamma x_i^{(S)}}
\end{align*}
\end{small}

\vspace{-4mm}
\begin{small}
\begin{align*}
\Longrightarrow
\frac{d R_i^{(S,\mathbf{x})} }{d x_i^{(S)}} 
&= \frac{(\gamma r - c_i)E_i^{(S,\mathbf{x}^{(S)})} - \gamma A_i^{(S,\mathbf{x})}}{\left(E_i^{(S,\mathbf{x}^{(S)})} +\gamma x_i^{(S)} \right)^2}
\end{align*}
\end{small}

The denominator is positive, while the numerator is a constant with respect to $x_i^{(S)}$, since $A_i^{(S,\mathbf{x})  }$ and $E_i^{(S,\mathbf{x}^{(S)})}$ do not depend on $x_i^{(S)}$.
So, $R_i^{(S,\mathbf{x})}$ is a monotone function of $x_i^{(S)}$, and whether it is increasing or decreasing, depends on the sign of the numerator: $(\gamma r - c_i)E_i^{(S,\mathbf{x}^{(S)})} - \gamma A_i^{(S,\mathbf{x})}$.
\end{proof}

\begin{proposition}
In MPE for Scenario 1, a player $i$ invests its maximal power if $\gamma r > c_i$, no power if $\gamma r < c_i$, and any amount of power if $\gamma r = c_i$.
\label{prop:sc1result}
\end{proposition}

\begin{proof}

Let $W^{(S,\mathbf{x})}$ be the  row  $\mathcal{O}(S)$ of $\mathbf{W}^{(\mathbf{x})}$.
Note that $A_i^{(S,\mathbf{x})} \!=\!    (E_i^{(S,\mathbf{x}^{(S)})} \!+ \gamma x_i^{(S)}) W^{(S,\mathbf{x})} \hat{R}_i^{(\mathbf{x})}$.
From the proof of \Cref{lem:monotone},
$\frac{d R_i^{(S,\mathbf{x})}}{d x_i^{(S)}}\!$ has  same sign as 
$(\gamma r - c_i) E_i^{(S,\mathbf{x}^{(S)})} \!- \gamma A_i^{(S,\mathbf{x})}\!$,
{which can  be written as}

~
\vspace{-3mm}
\begin{small}
\begin{align*}
& 
(\gamma r - c_i) E_i^{(S,\mathbf{x}^{(S)})} -  \gamma (E_i^{(S,\mathbf{x}^{(S)})} + \gamma x_i^{(S)})  W^{(S,\mathbf{x})} \hat{R}_i^{(\mathbf{x})}
\\
&= (\gamma r - c_i) E_i^{(S,\mathbf{x}^{(S)})} - \gamma (E_i^{(S,\mathbf{x}^{(S)})} + \gamma x_i^{(S)}) (\hat{R}_i^{(S,\mathbf{x})} - Z_i^{(S,\mathbf{x})}) 
\\
&= (\gamma r - c_i) E_i^{(S,\mathbf{x}^{(S)})} \!-\! \gamma \hat{R}_i^{(S,\mathbf{x})} (E_i^{(S,\mathbf{x}^{(S)})} \!+\! \gamma x_i^{(S)}) 
\\ &\;\;\;\;\;   
\!+\! \gamma \frac{(\gamma r - c_i) x_i^{(S)}}{E_i^{(S,\mathbf{x}^{(S)})} \!+\! \gamma x_i^{(S)}} (E_i^{(S,\mathbf{x}^{(S)})} \!+\! \gamma x_i^{(S)})
\\
&= (\gamma r \!-\! c_i) E_i^{(S,\mathbf{x}^{(S)})} \!-\! \gamma \hat{R}_i^{(S,\mathbf{x})} (E_i^{(S,\mathbf{x}^{(S)})} \!+\! \gamma x_i^{(S)}) \!+\! \gamma (\gamma r \!-\! c_i) x_i^{(S)}
\\
&= (\gamma r \!-\! c_i ) E_i^{(S,\mathbf{x}^{(S)})} \!-\! \gamma \hat{R}_i^{(S,\mathbf{x})} E_i^{(S,\mathbf{x}^{(S)})} \!+\! \gamma x_i^{(S)} (\gamma r \!-\! c_i \!-\! \gamma \hat{R}_i^{(S,\mathbf{x})})
\\
&= E_i^{(S,\mathbf{x}^{(S)})} (\gamma r -c_i  - \gamma \hat{R}_i^{(S,\mathbf{x})}) + \gamma x_i^{(S)} (\gamma r -c_i - \gamma \hat{R}_i^{(S,\mathbf{x})})
\\
&= (\gamma r -c_i  - \gamma \hat{R}_i^{(S,\mathbf{x})}) (E_i^{(S,\mathbf{x}^{(S)})} + \gamma x_i^{(S)})
\\
&= \gamma \Big(r - \frac{c_i}{\gamma} - \hat{R}_i^{(S,\mathbf{x})} \Big)(E_i^{(S,\mathbf{x}^{(S)})} + \gamma x_i^{(S)})
\end{align*} 
\end{small}

\noindent
Since $E_i^{(S,\mathbf{x}^{(S)})} + \gamma x_i^{(S)}$ is positive, and 
\mbox{$(r - \frac{c_i}{\gamma} - \hat{R}_i^{(S,\mathbf{x})} )$} has the same sign as  $(\gamma r -c_i)$
from \Cref{lem:bounded}, we have that $\frac{d R_i^{(S,\mathbf{x})}}{d x_i^{(S)}}$ has the same sign as $(\gamma r - c_i)$.
Also, note that if $\gamma r=c_i$, we have $R_i^{(S,\mathbf{x})}=0, \forall S \in 2^\mathcal{U}$  from Proposition~\ref{prop:convergence} when 
$\Gamma^{(S,\mathbf{x}^{(S)})} =  \gamma \left( \sum_{j \in S} x_j^{(S)} +\ell \right)$.

So, in any state $S$,
it is a dominant strategy for a player $i$ to invest its maximal power if $\gamma r > c_i$, no power if $\gamma r < c_i$, and any amount of power if $\gamma r = c_i$.
Since  the maximal power of a player $i$ would be  bounded (let the bound be $\overline{x}_i$), it would invest $\overline{x}_i$ if $\gamma r > c_i$.
Hence, we have a consistent solution for the Bellman equations that a player $i$ invests $\overline{x}_i$ if $\gamma r > c_i$, $0$ if $\gamma r < c_i$, and any amount of power in the range $[0,\overline{x}_i]$ if $\gamma r = c_i$.
\end{proof}

Thus,  the MPE strategy of a player follows a threshold policy, with a threshold on its cost parameter $c_i$ (whether it is lower than $\gamma r$) or alternatively, a {threshold on the offered reward $r$ (whether it is higher than $\frac{c_i}{\gamma}$).}
Note that though a player $i$ invests
maximal power when $\gamma r > c_i$, this is not inefficient since the power would be spent for less time as the problem would get solved faster.
An intuition behind this result is that, when there are several miners in the system, the competition drives  miners to invest heavily.
On the other hand, when there are few miners in the system, miners invest heavily so that  the problem gets solved faster (before  arrival of more competition).
Also, since the MPE strategies do not depend on $S$, the assumption of 
state knowledge can  be relaxed.

We now provide an intuition for why the MPE strategies are independent of the arrival and departure rates.
{From Proposition~\ref{prop:convergence}, $\mathbf{R}_i^{(\mathbf{x})} \!=\! (\mathbf{I}-\mathbf{W}^{(\mathbf{x})})^{-1} \mathbf{Z}_i^{(\mathbf{x})}$.
For $\gamma r>c_i$,} when power $x_i^{(S)}$ increases,  $\mathbf{Z}_i^{(\mathbf{x})}$ increases and $(\mathbf{I}-\mathbf{W}^{(\mathbf{x})})^{-1}$ decreases.
But 
$\mathbf{R}_i^{(\mathbf{x})}$ increases with $x_i^{(S)}$ when $\gamma r > c_i$ (shown in the proof of Proposition~\ref{prop:sc1result}), implying that the rate of increase of $\mathbf{Z}_i^{(\mathbf{x})}$ dominates the rate of decrease of $(\mathbf{I}-\mathbf{W}^{(\mathbf{x})})^{-1}$.
So, the effect of $\mathbf{W}^{(\mathbf{x})}$ and hence  state transitions is relatively weak, thus resulting in Markovian players playing  strategies that are independent of the arrival and departure rates. Similar argument holds for $\gamma r \leq c_i$.
It would be interesting  to study scenarios where the 
rate of problem getting solved
is 
a non-linear function of the players' invested powers. While a linear function is suited to most  distributed computing applications,
a non-linear function could possibly see $\mathbf{W}^{(\mathbf{x})}$ having a strong effect leading to MPE being dependent on the arrival and departure rates.

For analyzing the expected utility of a strategic player $j$, let us consider that the power available to it is very large, say $\overline{x}_j$.
Following our result on MPE, every player $j$ satisfying $c_j < \gamma r$ would invest $\overline{x}_j$ entirely.
So, 
we  have that 
$\gamma(\sum_{j\in S, c_j < \gamma r} \overline{x}_j + \ell)$
is very large, and hence $D^{(S,\mathbf{x})}$ (which now approximates to $\gamma(\sum_{j\in S, c_j < \gamma r} \overline{x}_j + \ell)$) is also very large. 
Since  we know that $R_i^{(S\cup \{j\},\mathbf{x})}$ and $R_i^{(S\setminus \{j\},\mathbf{x})}$ are bounded by a small quantity
from \Cref{lem:bounded}, the limit of the expected utility $R_i^{(S,\mathbf{x})}\!$ computed in any state $S$ (from Equation~(\ref{eqn:recursive_Sc1})) is 
$\frac{\overline{x}_i}{\sum_{j\in S, c_j < \gamma r} \overline{x}_j + \ell}r - \frac{c_i \overline{x}_i}{\gamma (\sum_{j\in S, c_j < \gamma r} \overline{x}_j + \ell)}$.
To get further insight into this, 
say $\ell$ is considered insignificant, that is, the computation is dominated by strategic players. 
Furthermore, say for every strategic player $i$, $c_i < \gamma r$, and let the very large amount of power available to these players be the same ($\overline{x}_i = \overline{x}_j, \forall i,j \in \mathcal{U}$).
Thus, the limit of the expected utility $R_i^{(S,\mathbf{x})}\!$ computed in any state $S$ simplifies to $\frac{r}{|S|} - \frac{c_i}{\gamma |S|}$.
This is intuitive, since if $\overline{x}_i = \overline{x}_j, \forall i,j \in \mathcal{U}$, the reward would be won by the players with equal probability (hence the term $\frac{r}{|S|}$), and the cost is reduced owing to the reduced time due to the combined rate of the problem getting solved (hence the term $\frac{c_i}{\gamma |S|}$).

\vspace{-2mm}
\section{Scenario 2: Analysis of MPE
}
\label{sec:scenario2}

\begin{proposition}
In MPE for Scenario 2, a player $i$ invests 

\vspace{-2mm}
\begin{small}
\begin{align*}
x_i^{(S)} 
&= \max \left\{ \psi^{(S)} \Big( 1 - \frac{c_i \psi^{(S)} }{  r \beta} \Big) , 0 \right\}
\end{align*}
\end{small}

\noindent
where 
$
\psi^{(S)} = r \beta \frac{|\hat{S}|-1 + \sqrt{ (|\hat{S}|-1)^2 + \frac{4\ell}{r \beta}\!  \sum_{j \in \hat{S}} c_j }}{2  \sum_{j \in \hat{S}} c_j}
$.
Here, $\hat{S}$ is the maximal set of players $j \in S$ which collectively satisfy the constraints $c_j < \frac{r \beta}{\psi^{(S)}}$.
\label{prop:sc2result}
\end{proposition}

\begin{proof}

Recall that since $\mathbf{W}^{(\mathbf{x})}$ is a strictly substochastic matrix, 
\mbox{$(\mathbf{I}-\mathbf{W}^{(\mathbf{x})})^{-1} = \lim_{t \rightarrow \infty} \sum_{\eta=0}^{t-1} (\mathbf{W}^{(\mathbf{x})})^\eta$}.
Since all the elements of $\mathbf{W}^{(\mathbf{x})}$ are non-negative, all the elements of $( \mathbf{W}^{(\mathbf{x})} )^\eta$ also are non-negative for any natural number $\eta$, and hence all the elements of $(\mathbf{I}-\mathbf{W}^{(\mathbf{x})})^{-1}$ are non-negative.
Also, since $\mathbf{R}_i^{(\mathbf{x})} \!=\! (\mathbf{I}-\mathbf{W}^{(\mathbf{x})})^{-1} \mathbf{Z}_i^{(\mathbf{x})}$ (Proposition~\ref{prop:convergence}) and since $\mathbf{W}^{(\mathbf{x})}$ is independent of $x_i^{(S)}$ in this scenario, 
maximizing the components of $\mathbf{Z}_i^{(\mathbf{x})}$ (namely, $Z_i^{(S,\mathbf{x})}$) individually with respect to $x_i^{(S)}$ would  essentially maximize all the elements of $\mathbf{R}_i^{(\mathbf{x})}$.
Now, since $\Gamma^{(S,\mathbf{x}^{(S)})} = \beta$ in this scenario, we have

\vspace{-2mm}
\begin{small}
\begin{align*}
& Z_i^{(S,\mathbf{x})} \!=\! \Bigg( \frac{\beta}{\sum_{j\in S} x_j^{(S)}+\ell} r \!-\!c_i \Bigg)\frac{x_i^{(S)}}{D^{(S,\mathbf{x})}} 
\end{align*}
\end{small} 

\noindent
As $D^{(S,\mathbf{x})}$ is independent of $x_i^{(S)}$ in this scenario, 
it can be  shown that $Z_i^{(S,\mathbf{x})}$ is a concave function with respect to $x_i^{(S)}$ (the second derivative turns out to be $\frac{-2 r \ell \beta}{(\sum_{j\in S} x_j^{(S)} +\ell)^3 D^{(S,\mathbf{x})}}$).
The first order condition $\frac{d Z_i^{(S,\mathbf{x})}}{d x_i^{(S)} }=0$ gives

\vspace{-2mm}
\begin{small}
\begin{align*}
x_i^{(S)} &=  
\Bigg( \sum_{j\in S} x_j^{(S)} + \ell \Bigg) \Bigg( 1 - \frac{c_i}{ r \beta } \Big( \sum_{j\in S}  x_j^{(S)} +\ell \Big) \Bigg)
\end{align*}
\end{small}

\noindent
Let $\psi^{(S)} \!=\! \sum_{j \in {S}} x_j^{(S)} + \ell$. 
As $x_i^{(S)}$ is non-negative, we have

\vspace{-2mm}
\begin{small}
\begin{align}
x_i^{(S)} &=  
\max \left\{ \psi^{(S)} \Big( 1 - \frac{\psi^{(S)}}{r \beta} c_i \Big) , 0 \right\}
\label{eqn:SNE_Sc2}
\end{align}
\end{small}

\noindent
Let $\hat{S} = \{ j \in S : x_j^{(S)} > 0 \}$. We later show how to determine  set $\hat{S}$.
Summing the above over all players in $S$ and then adding $\ell$ on both sides, we get

~
\vspace{-3mm}
\begin{small}
\begin{align*}
\sum_{j \in S}  x_j^{(S)} + \ell = \psi^{(S)} \Bigg(  |\hat{S}| -  \frac{\psi^{(S)}}{r \beta} \sum_{j \in \hat{S}} c_j \Bigg) + \ell
\end{align*}
\end{small}

\noindent
Substituting  $\sum_{j \in S}  x_j^{(S)} + \ell$ as  $\psi^{(S)}$, we get

\begin{small}
\begin{align*}
\frac{1}{r \beta} \sum_{j \in \hat{S}} c_j \left( \psi^{(S)} \right)^2 - (|\hat{S}|-1) \psi^{(S)} - \ell = 0
\end{align*}
\end{small}

\noindent
Solving this equation for positive value of $\psi^{(S)}$, we get

\begin{small}
\begin{align*}
\psi^{(S)} = r \beta \frac{|\hat{S}|-1 + \sqrt{ (|\hat{S}|-1)^2 + \frac{4\ell}{r \beta}  \!  \sum_{j \in \hat{S}} c_j }}{2  \sum_{j \in \hat{S}} c_j}
\end{align*}
\end{small}

\noindent
Substituting this expression for $\psi^{(S)}$ in Equation~(\ref{eqn:SNE_Sc2}) gives the MPE strategy of player $i$.
\end{proof}

\vspace{2mm}
So, $x_i^{(S)} > 0$ iff $c_i < \frac{2  \sum_{j \in \hat{S}} c_j}{|\hat{S}|-1 + \sqrt{ (|\hat{S}|-1)^2 + \frac{4\ell}{r \beta}\!  \sum_{j \in \hat{S}} c_j }} $.
That is,
only players with cost parameters in a relatively low range in a given state, invest.
The constraint implies that if player $i$ invests, then player $j$ with $c_j < c_i$  also invests.
So, there exists a threshold player $\hat{i}$ such that any player $j$ with $c_j > c_{\hat{i}}$ would not invest.
Hence, set $\hat{S}$ can be constructed iteratively (initiating from an empty set) by adding players $j$ from set $S \setminus \hat{S}$ one at a time, in ascending order of $c_j$, until the above constraint is violated for the cost parameter of the newly added player.

To get a better understanding of this result, if the power $\ell$ invested by fixed players is considered insignificant, we have  $\psi^{(S)} = r \beta \frac{|\hat{S}|-1}{\sum_{j\in \hat{S}} c_j}$ and  the condition for $x_i^{(S)}>0$ simplifies to $c_i < \frac{\sum_{j\in \hat{S}} c_j}{|\hat{S}|-1}$.

Furthermore, if the strategic players are homogeneous ($c_i=c_j, \forall i,j \in \mathcal{U}$), the cost constraint is satisfied for all players in $S$ (since $c < \frac{|S|c}{|S|-1}$) and so, all the strategic players invest $\frac{r \beta}{c}\big( \frac{|S|-1}{|S|^2} \big) $. That is, if the computation is dominated by strategic players which are homogeneous, they would invest proportionally to the `reward to cost parameter' ratio in MPE.

Since the transition probabilities, and hence $\mathbf{W}^{(\mathbf{x})}$, are constant w.r.t. players' strategies in this scenario,
a player's MPE utility computed in state $S$ ($R_i^{(S,\mathbf{x})}$) is a linear combination (with constant non-negative weights) of its utilities over all states computed without accounting for state transitions.
Hence, the MPE strategies are independent of the arrival and departure rates.

Note that while the decision regarding whether or not to invest was independent of the cost parameters of the other players in the system in Scenario 1, this decision highly depends on the cost parameters of other players in Scenario 2.

\vspace{-2mm}
\section{Simulation Study
}

Throughout the paper, we determined   MPE strategies, which we observed to be independent of players' arrival and departure rates. 
However, it is clear from Equations~(\ref{eqn:recursive}), (\ref{eqn:recursive_Sc1}), (\ref{eqn:recursive_Sc2}) and \Cref{prop:convergence} that the players' utilities would indeed depend on these rates.
We now study the effects of these rates on the  utilities in MPE using simulations.
In order to reliably obtain an accurate relation between the arrival/departure rates and the expected utilities of the players, we consider that the computation is dominated by the strategic players (that is, the power invested by the fixed players is insignificant: $\ell \rightarrow 0$) and the strategic players are homogeneous (their arrival/departure rates and their cost parameters are the same).
Let $\lambda,\mu,c$ denote the common arrival rate, departure rate, and cost parameter, respectively.
Note that if the strategic players are considered homogeneous, the players' sets (states) can be mapped to their cardinalities.
We observe how the expected utility of a player changes as a function of the number of other players present in the system, for different arrival/departure rates.
In our simulations, we consider the following values:
$r =  10^5, \gamma = \beta = 0.1, |\mathcal{U}| = 10^4, c = 0.003$  (a  justification of
these values is provided in Appendix).

\vspace{1mm}
\noindent
\textbf{Statewise Nash Equilibrium.}
For a comparative study, we also look at the equilibrium strategy profile of a given set of players $S$, when there are no arrivals and departures ($\lambda_j=0, \forall j \notin S$ and $\mu_j=0, \forall j \in S$). 
We call this, {\em statewise\/} Nash equilibrium (SNE) in state $S$.
Since the MPE strategies of the players are independent of the arrival and departure rates, a player's SNE strategy in a state is same as its MPE strategy corresponding to that state.
Note, however, that the expected utilities in SNE would be different from those in MPE, since the expected utilities highly depend on the arrival and departure rates (Equations~(\ref{eqn:recursive}), (\ref{eqn:recursive_Sc1}), (\ref{eqn:recursive_Sc2}) and \Cref{prop:convergence}).
Also, since SNE does not account for  change of the set of players present in the system,
the expected utilities 
in SNE
for different values on X-axis in the plots are computed independently of each other.

\vspace{-2mm}
\subsection{Simulation Results
}

In Figures \ref{fig:plots_scenario1} and \ref{fig:plots_scenario2},
the plots for  expected utility largely follow near-linear curve (of  negative slope) on  log-log scale, with respect to the number of players in the system.
That is, they nearly follow power law, which means that scaling the number of players by a constant factor would lead to proportionate scaling of  expected utility.

\vspace{1mm}
\noindent
\textbf{Scenario 1.}
\Cref{fig:plots_scenario1} presents  plots for expected utilities with   MPE policy for various values of $\lambda$ and $\mu$, and compares them with  expected utilities in SNE.
Following are some  insights:

\begin{itemize}[leftmargin=*]
\setlength\itemsep{0em}

\item
As  seen at the end of  \Cref{sec:scenario1}, 
if the mining is dominated by strategic players which are homogeneous,
the expected utilities in MPE  are bounded by $\frac{r}{|S|} - \frac{c}{\gamma|S|}$.
It can  be similarly shown that the limit of the players' expected utilities in SNE is $\frac{r}{|S|} - \frac{c}{\gamma|S|}$ (this can be seen by substituting in Equation~(\ref{eqn:recursive_Sc1}): $\lambda_j=0 \; \forall j \notin S$, $\mu_j = 0$, $c_j = c$, $x_j^{(S)} \rightarrow \infty, \forall j \in S$, and $\ell \rightarrow 0$).
Owing to this, the expected utilities in MPE are bounded by the expected utilities in SNE, which is reflected in \Cref{fig:plots_scenario1}.

\item
In Scenario 1, a higher $\lambda$ results in a higher likelihood of the system having more players, which results in a higher rate of the problem getting solved as well as more competition. 
This, in turn, reduces the time spent in the system as well as the probability of winning for each player, which hence reduces the cost incurred as well as the expected reward. 
\Cref{fig:plots_scenario1}(a) suggests that, as $\lambda$ changes, the change in cost incurred balances with the change in expected reward, since the change in expected utility is insignificant.

\item
For a given $\mu$, if the number of players changes, there is a balanced tradeoff between the cost  and the expected reward as above; so the change in expected utility is insignificant.
But a higher $\mu$ results in a higher probability of player $i$ departing from the system and staying out when the problem gets solved, thus lowering its expected utility (\Cref{fig:plots_scenario1}(b)).

\end{itemize}

\noindent
\textbf{Scenario 2.}
Since a player's SNE strategy in a state is same as its MPE strategy corresponding to that state, 
a player's SNE startegy is to invest
$\frac{r \beta}{c}\big( \frac{|S|-1}{|S|^2} \big)$ in state $S$
(as explained at the end of \Cref{sec:scenario2} when  computation is dominated by strategic players that are homogeneous).
Furthermore, in SNE, the expected utility of each player can be shown to be $\frac{r}{|S|^2}$ in state $S$ (this can be seen by substituting in Equation~(\ref{eqn:recursive_Sc2}): $\lambda_j=0 \; \forall j \notin S$, $\mu_j = 0$, $c_j = c$, $x_j^{(S)} = \frac{r \beta}{c}\big( \frac{|S|-1}{|S|^2} \big), \forall j \in S$, and $\ell \rightarrow 0$).
\Cref{fig:plots_scenario2} presents the plots for expected utilities with the analyzed MPE policy for different values of $\lambda$ and $\mu$, and compares them against 
SNE.
Following are some insights:

\begin{figure}
\begin{tabular}{cc}
\hspace{-4mm}
\includegraphics[width=.27\textwidth]{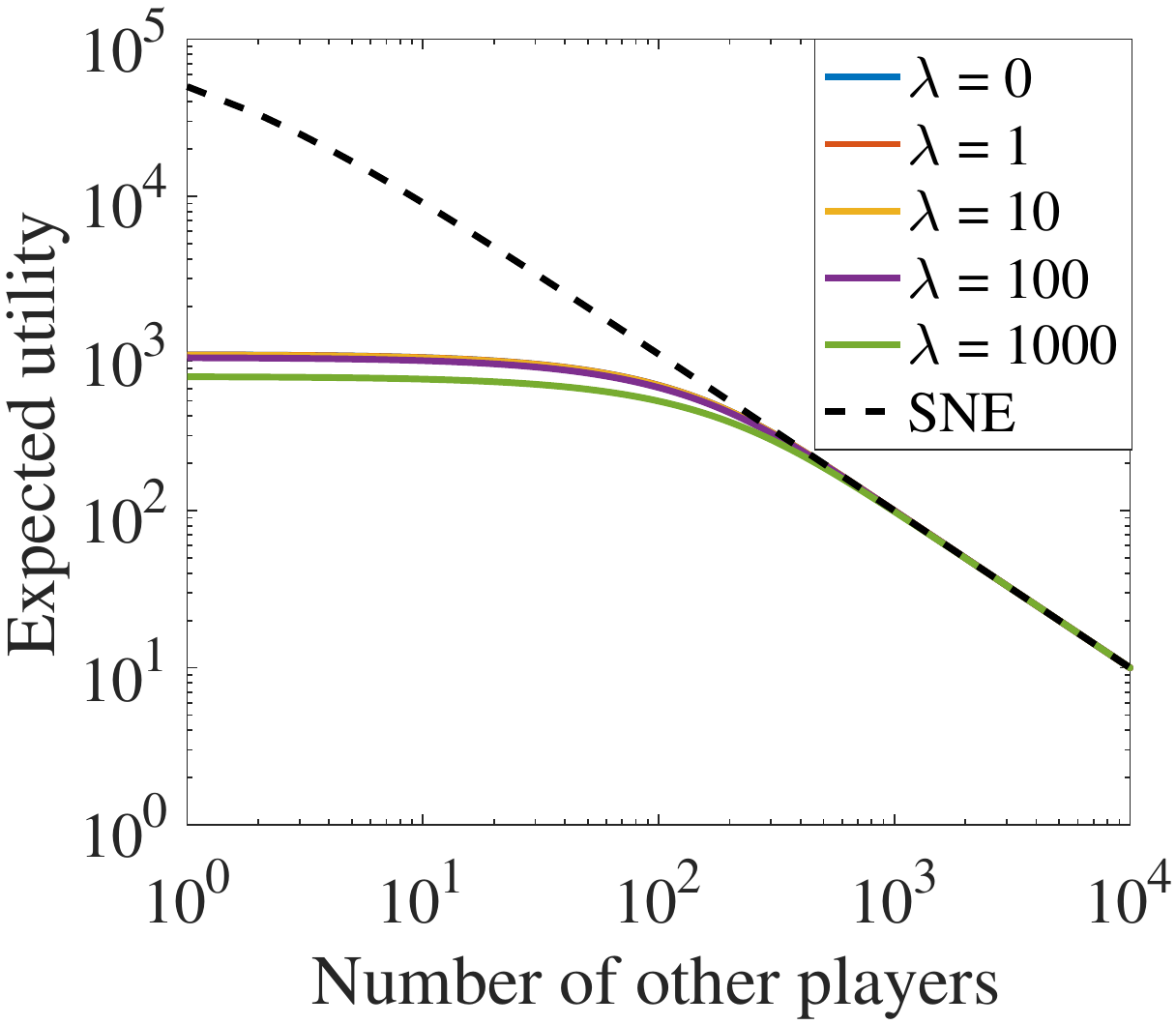}
&
\hspace{-2mm}
\includegraphics[width=.27\textwidth]{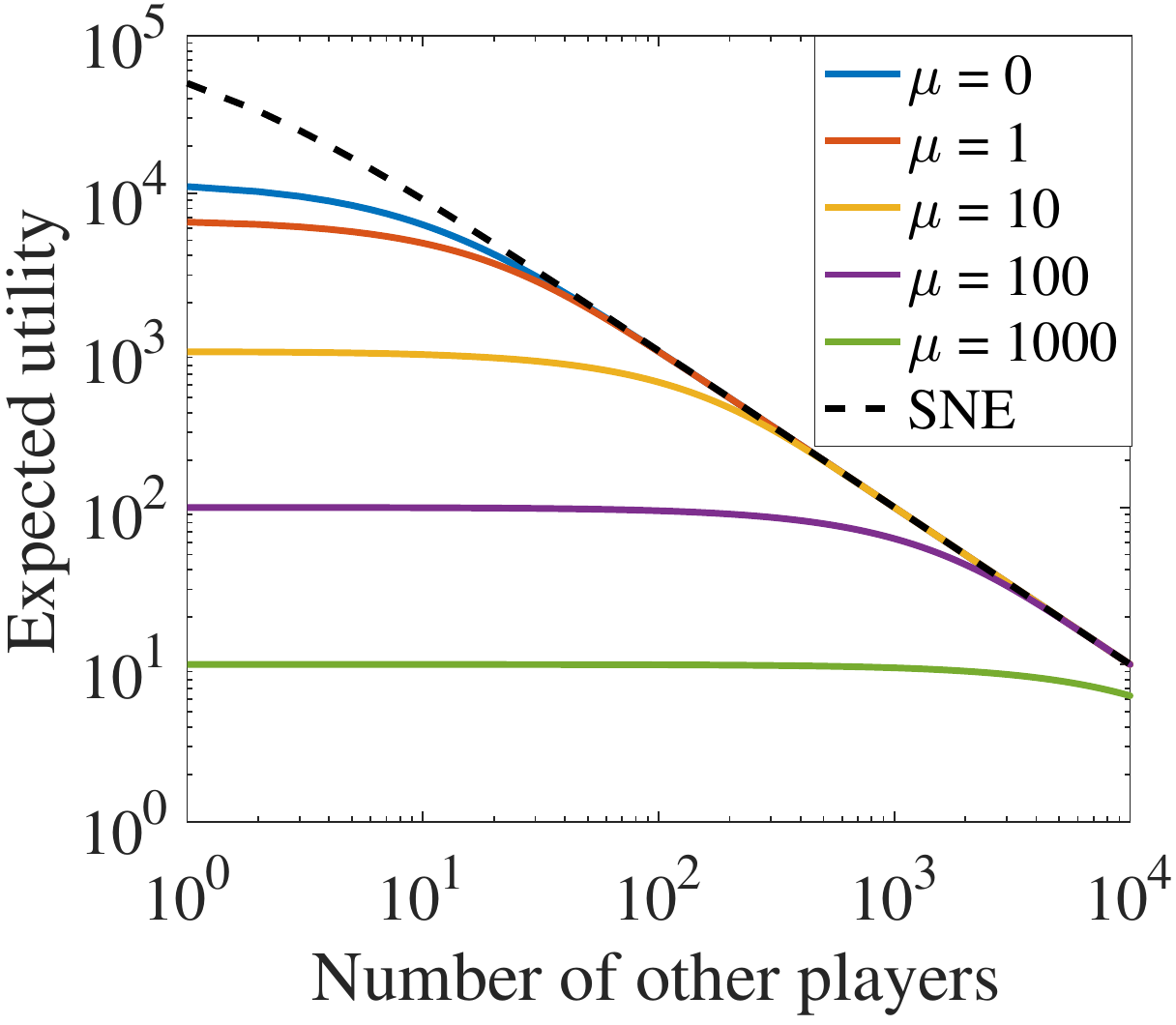}
\\
(a) for different $\lambda$'s ($\mu=10$)
&
(b) for different $\mu$'s ($\lambda=10$)
\end{tabular}
\caption{Expected utility of a player in Scenario 1}
\label{fig:plots_scenario1}
\end{figure}

\begin{figure}
\begin{tabular}{cc}
\hspace{-4mm}
\includegraphics[width=.27\textwidth]{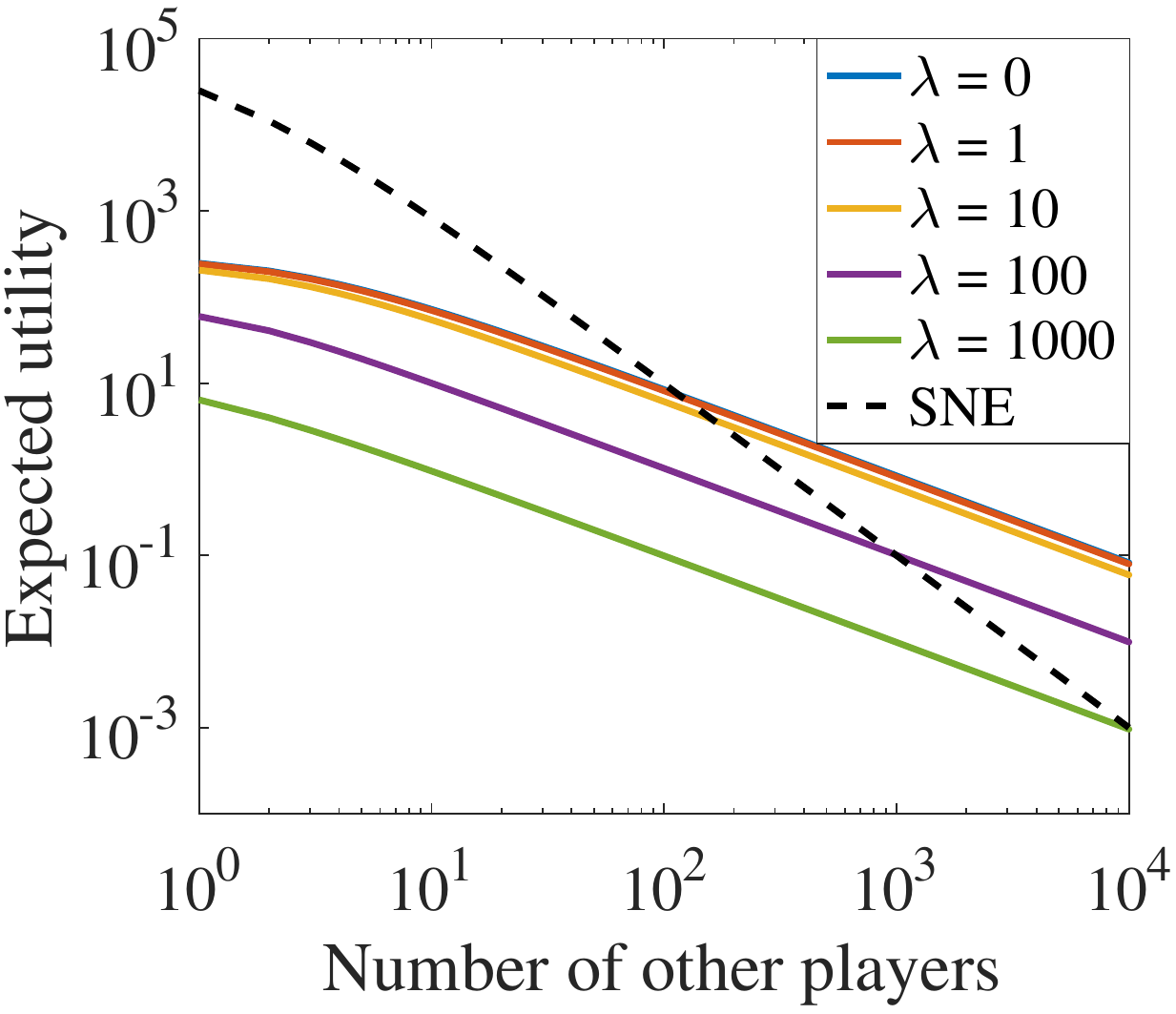}
&
\hspace{-2mm}
\includegraphics[width=.27\textwidth]{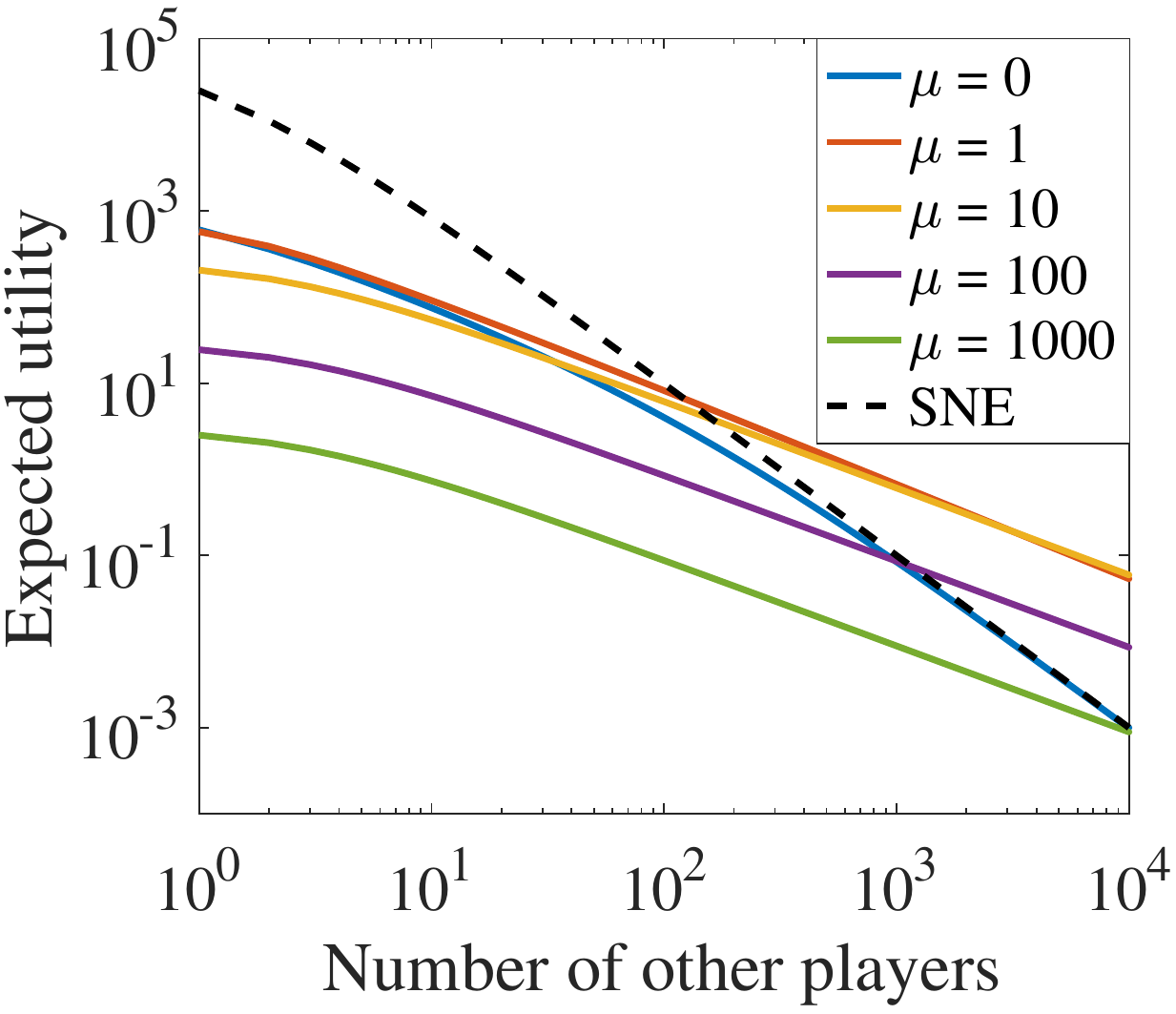}
\\
(a) for different $\lambda$'s ($\mu=10$)
&
(b) for different $\mu$'s ($\lambda=10$)
\end{tabular}
\caption{Expected utility of a player in Scenario 2 
} 
\label{fig:plots_scenario2}
\end{figure}

\begin{itemize}[leftmargin=*]
\setlength\itemsep{0em}

\item
An increase in the number of players increases  competition for the offered reward and hence reduces the reward per unit time received by each player,
with no balancing factor (unlike in Scenario 1);
so  
the expected utility decreases.

\item
For  higher $\lambda$, there is  higher likelihood of  system having more players,  thus resulting in  lower expected utility owing to  aforementioned reason. 
Also, from \Cref{fig:plots_scenario2}(a), if $\lambda$ is not very high, an increase in $\mu$ is likely to reduce the competition to the extent that the expected MPE utility 
when the number of players in the system is large, can exceed the corresponding SNE utility ($\frac{r}{|S|^2}$, which would be very low when the number of players in the system is large).

\item
A higher $\mu$ likely results in less competition, however it also results in a higher probability of player $i$ departing from the system and hence losing out on the reward for the time it stays out;
this leads to a tradeoff.
\Cref{fig:plots_scenario2}(b) shows that the effect of the  probability of player $i$  departing from the system dominates 
the effect of  
the reduction in competition.
For similar reasons as above, the expected MPE utility 
when the number of players in the system is large,
can exceed the corresponding SNE utility.
\end{itemize}

\vspace{-2mm}
\section{Future Work
}

One could study a variant of Scenario 1 where the rate of problem getting solved (and perhaps also the cost) increases non-linearly with  the invested power.
Since players are seldom completely rational in real world,
it would be useful to study the game under bounded rationality.
To develop a more sophisticated stochastic model,
one could obtain real data concerning the arrivals and departures of players and their investment strategies.
Another promising possibility is to incorporate state-learning in our model.
A Stackelberg game could be studied, where the system decides the amount of reward to offer, and then the computational providers decide how much power to invest based on the offered reward.

\vspace{-2mm}
\section*{Appendix
}

We take cues from bitcoin mining for our numerical simulations.
The current offered reward 
for successfully mining a block is $12.5$ bitcoins. Assuming 
1 bitcoin $\approx$ $\$ 8000$, the reward translates to $\$10^5$.
The bitcoin problem complexity is  set such that it takes around $10$ minutes on average for a block to get mined. That is, the rate of problem getting solved is 
$0.1$ per minute on average.
One of the most powerful ASIC (application-specific integrated circuit) currently available in  market is Antminer S9, which performs computations of upto 13 TeraHashes per sec, while 
consuming about $1.5$ kWh in 1 hour, which translates to $\$0.18$ per hour (at the rate of $\$0.12$ per kWh), equivalently $\$0.003$ per minute.
As per BitNode (\url{bitnodes.earn.com}), a crawler developed to estimate the size of  bitcoin network,
the number of bitcoin miners is 
around $10^4$.
Hence, we consider
$r =  10^5, \gamma = \beta = 0.1, c = 0.003, |\mathcal{U}| = 10^4$.

\vspace{-2mm}
\section*{Acknowledgment}

The work is partly supported by CEFIPRA grant No. IFC/DST-Inria-2016-01/448 ``Machine Learning for Network Analytics''. 


\balance

\begin{small}
\bibliographystyle{IEEEtran}

\bibliography{Blockchain_references}
\end{small}

\end{document}